\documentclass[preprint]{elsarticle}
\usepackage{amsthm} 

\newtheorem{lem}{Lemma}

\newtheorem{propos}{Proposition}

\newtheorem{rem}{Remark}

\usepackage{comment}
\usepackage{soul}
\usepackage{amsmath}
\usepackage{algorithm}
\usepackage[noend]{algpseudocode}
\usepackage{multirow,tabularx}
\usepackage{color}
\usepackage{graphicx}
\usepackage{amsmath}
\usepackage{amssymb}
\usepackage{multicol}
\usepackage{adjustbox}
\usepackage[graphicx]{realboxes}
\usepackage{subcaption} 
\usepackage{rotating} 
\usepackage{flushend}
\usepackage{natbib}
\usepackage{lineno}
\usepackage{adjustbox}
\usepackage[utf8]{inputenc}
\usepackage{hyperref}
\hypersetup{colorlinks=true,linkcolor=blue,filecolor=magenta,urlcolor=cyan}

\makeatletter
\def\ps@pprintTitle{%
 \let\@oddhead\@empty
 \let\@evenhead\@empty
 \def\@oddfoot{}%
 \let\@evenfoot\@oddfoot}
\makeatother

\let\rm\mathrm

\begin{document}
\begin{frontmatter}

\title{\LARGE \bf
Including steady-state information in nonlinear models: an application to the development of soft-sensors}

\author[]{Leandro Freitas$^{a}$}
\ead{leandro.freitas@ifmg.edu.br}
\author[]{Bruno H. G. Barbosa$^{b}$}
\ead{brunohb@ufla.br}
\author[]{Luis A. Aguirre$^{c}$}
\ead{aguirre@ufmg.br}

\address{$^{a}$Departamento de Automa\c{c}\~{a}o e Inform\'{a}tica, Instituto Federal de Educa\c{c}\~{a}o, Ci\^{e}ncia e Tecnologia de Minas Gerais Campus Betim, 32677-564, Betim, MG, Brazil} 
\address{$^{b}$ Departamento de Autom\'{a}tica, Universidade Federal de Lavras, CP 3037, 37200-000, Lavras, MG, Brazil} 
\address{$^{c}$ Departamento de Engenharia Eletr\^{o}nica, Universidade Federal de Minas Gerais - Av. Ant\^{o}nio Carlos 6627, 31270-901, Belo Horizonte, MG, Brazil} 

\begin{abstract}
When the dynamical data of a system only convey dynamic information over a limited operating range, the identification of models with good performance over a wider operating range is very unlikely. 
To overcome such a shortcoming, this paper describes a methodology to train models from dynamical data {\it and}\, steady-state information, which is assumed available. The novelty is that the procedure can be applied to models with rather complex structures such as multilayer perceptron neural networks in a bi-objective fashion without the need to compute fixed points neither analytically nor numerically. As a consequence, the required computing time is greatly reduced. The capabilities of the proposed method are explored in numerical examples and the development of soft-sensors for downhole pressure estimation for a real deep-water offshore oil well. The results indicate that the procedure yields suitable soft-sensors with good dynamical {\it and}\, static performance and, in the case of models that are nonlinear in the parameters, the gain in computation time is about three orders of magnitude considering existing approaches.

\end{abstract}


\end{frontmatter}

\section{Introduction}


In deep water gas-lift  oil well processes \cite{Elldakli2017}, the downhole pressure is an important variable to ensure safety and to provide useful information for management and oil recovery of the oil field \cite{camponogara2010}. It can be used by anti-slugging control systems or to optimize the costs of the oil production. This pressure is measured by the Permanent Downhole Gauge (PDG). However, due to extremely hostile operating conditions PDGs, which could be located at depths greater than four kilometers, often stop working \cite{morais2019meas}. The maintenance of the sensor is not economically viable and involves high environmental risks.

In this wise, the development of soft-sensors to estimate the downhole pressure has become a valuable alternative to provide this important information \cite{aguirre2017cep}. Basically, a soft-sensor is a predictive mathematical model that estimates some quantity based on measurements of other process variables \cite{ALQUTAMI201872, kadlec2009}. In the case of the PDG, the estimated pressure is normally based on measures from the platform, where lower measurement uncertainties are expected or, even at the \emph{wet christmas tree} \cite{teixeira2014}.







Since several physical aspects are involved in the process, a complete phenomenological modeling of the downhole pressure is very difficult, this makes the use of \emph{system identification} tools prominent. The main purpose of \emph{system identification} is to build dynamic models based on experimental data. To this end, related fields of knowledge like \emph{statistics}, \emph{optimization}, \emph{machine learning}, have become important tools to extract information about system dynamics from data. The use of artificial neural networks (ANN) for systems identification is another successful tool \cite{chen1990}, especially dealing with nonlinear systems. In most of cases, model parameters are estimated (network training) using just one source of information: the dynamical data set. This will be referred to as \emph{black box} identification.

However, traditional methods for developing soft-sensors do not efficiently use all the information hidden in process variables \cite{He2019}.
Determining a nonlinear model from a finite set of observations without any \emph{prior knowledge} about the system is an ill-posed problem  \cite{johansen1996}, because a unique model may not exist, or it may not depend continuously on the observations \cite{tikhonov1977}. This issue is worsened when dealing with noisy signals, non-informative data (e.g. non-persistently exciting inputs) and high-dimensional systems. From an optimization point of view, the search space and the number of local minima grow indefinitely, generating an extra challenge. Hence, physical insights about the system is almost a requirement to achieve suitable models. When no prior knowledge is available (e.g. black-box approach), it is common to assume that the system has some smoothness property, using a regularization term \cite{johansen1996, girosi1995}.

ANN are commonly used in dynamic systems identification mainly due to their capability of fitting data generally well. In this respect \emph{black-box} training is the rule, and in principle the ANN incorporates dynamical information about system behavior from the dynamical data set \cite{rib_agu/18}. If such a data set is sufficiently informative, the model is usually able to represent the system in various respects. An important practical shortcoming, especially in the case when the dynamical data are obtained from historical data, has to be faced when the available data are not sufficiently informative. In such instances some important aspects of the system may not be correctly incorporated by the model. Not only that, because of such a lack, even the information that is present in the dynamical data may not be correctly learned by the model \cite{He2019}.

One way of circumventing the aforementioned difficulty is to provide relevant missing information -- called \emph{auxiliary information} -- apart from the dynamical data set. Auxiliary information can be as general as the overall shape of static nonlinearity underlying the process \cite{aguirre2000buck}
or symmetry properties of the system \cite{aguirre2004b,chen2008}, or could be as specific as the precise static nonlinearity \cite{aguirre2004,aguirre2007}. It has been argued that the proper use of auxiliary information in the training of neural networks is beneficial in many practical ways \cite{joerding1991}. The use of auxiliary information is the central feature of \emph{grey-box} approaches \cite{Aguirre_2019} which have been used in several applications, including chemical processes \cite{thompson1994}, hydraulic systems \cite{barbosa2011,kilic2014}, energy systems \cite{sanchez2014}, fault detection \cite{cen2013}, chaotic systems \cite{nepomuceno2003} and oil industry \cite{aguirre2017cep} to mention but a few. 

In this paper the \emph{auxiliary information} is assumed to be a set of steady-state data, although other alternatives exist \cite{tulleken1993,eskinat1993,aguirre2004,nepomuceno2007,chen2011,Aguirre_2019}. Regarding the use of steady-state data in system identification \cite{nepomuceno2007,barroso2007,barbosa2011}, the approaches presented in \cite{nepomuceno2007,barroso2007} are only applied to linear-in-the-parameters models and the one presented in \cite{barbosa2011} is very computational demanding.

Thus, this work describes a novel grey-box identification strategy to include auxiliary information about static information in dynamic models. The method has a much lower computational cost than the one in \cite{barbosa2011} and it is applicable to a wide class of model structures, from polynomial to neural network models. The inclusion of auxiliary information is accomplished redefining the objective function (including penalty terms or adding new objectives) during parameter estimation. The main contribution of this work is to show that dynamical models can be identified using auxiliary static information without computing the models fixed points, which makes the proposed approach very suitable to identify rather complex models not imposing a heavy computational burden. Besides, the use of this auxiliary information helps to find models with better dynamical performance on operating regimes not originally represented in the dynamical dataset, as shown in the numerical results. 

This paper is organized as follows. Some of the main aspects of grey-box identification and some related works are briefly mentioned in Section~\ref{sec:Related}. The problem is defined in Section~\ref{sec:probStat} and a background is provided in Section~\ref{sec:backGr}. In Section~\ref{sec:methods} the proposed procedure is presented and the results and discussions are provided in Section~\ref{sec:results}. Section~\ref{sec:conclusions} shows the main conclusions and suggestions for future work.

\section{Grey-box Identification and Related Works}
\label{sec:Related}

A central challenge in grey-box identification is how to efficiently employ auxiliary information during training.
This can be done in a number of ways by means of constraints, linguistic rules and others \cite{sanchez2014, WU202074}. Hence it is convenient to distinguish grey-box strategies in terms of \emph{model class} (model structure), \emph{type of auxiliary information} (how it is expressed), and \emph{incorporating strategy} (the manner of including auxiliary information in the model).

In terms of \emph{model class}, grey-box identification was first implemented using linear structures \cite{tulleken1993,eskinat1993,johansen1996}, but it seems more powerful for nonlinear structures, including polynomial models \cite{correa2002,nepomuceno2003,aguirre2004,barbosa2011}, radial basis functions (RBF) \cite{aguirre2007,chen2009,chen2011,rego2014}, fuzzy systems \cite{abonyi2001,abdelazim2005,sanchez2014}, multilayer perceptron (MLP) or recurrent (RNN) neural networks \cite{psichogios1992,thompson1994,braake1998,oussar2001,aguirre2004b,WU202074}. 

Procedures for including auxiliary information in MLP and RNN networks seem to be less explored since they are nonlinear in the parameters, although some methods have been put forward for
semi-physical modeling \cite{thompson1994,braake1998,oussar2001,WU202074}, for static
information \cite{amaral2001} (exact matching) and symmetry \cite{aguirre2004b} for MLPs and RBF models. Models that are linear with respect to the parameters usually lead to convex problems, that are easier to deal with and the static curve can be sometimes determined analytically depending on the model structure \cite{aguirre2004}.

Barbosa and colleagues \cite{barbosa2011} described a method to include auxiliary information using bi-objective parameter estimation, where one objective is to improve the fitness to empirical data and the other is to improve the fitness to the auxiliary information. This approach was implemented using polynomial models and compared to other techniques (e.g. constrained polynomial models, ANNs). The main drawback of this method is the high computational cost, due to the calculation of the model static curve needed to evaluate the objective function (free run simulation over different operating points), and the use of evolutionary algorithms to solve the resulting nonconvex problem. 

In terms of the \emph{type of auxiliary information}, some approaches used the stability and sign of the stationary gain for linear models \cite{tulleken1993}, the phase crossover frequency \cite{eskinat1993}, the steady-state balance equations \cite{eskinat1993}, the steady-state values \cite{nepomuceno2007}, the static curve \cite{barbosa2011}, the symmetry \cite{aguirre2004}. 

There are several ways of \emph{incorporating} auxiliary information such as Bayesian approaches \cite{herbet1993}, nonlinear optimization techniques \cite{correa2002}, the constrained least squares algorithm \cite{aguirre2004}, or multiobjective optimization procedures \cite{johansen1996,nepomuceno2007,barbosa2011}. The type of auxiliary information available and the
model class are determining factors in the choice of the method to be used.

\section{Context and Problem Statement}
\label{sec:probStat}

Consider the following NARX (Nonlinear AutoRegressive with eXogenous inputs) model class used in this work
\begin{equation}
  \label{eq:model}
      y(k) = F\left( \boldsymbol{\psi}(k-1),\,\boldsymbol{\theta} \right),
\end{equation}
\sloppy where $k$ is the sample index, $F(\, .\,)$ is a nonlinear function, $\boldsymbol{\theta}\in\mathbb{R}^q$ is the vector of parameters to be estimated from measured data, $\boldsymbol{\psi}(k-1)=\left[1 \,\,\,\,\,\,  y(k-1)\dots y(k-n_y)\,\,\,\,\,\, u(k-1) \dots  u(k-n_u) \right]^T$ is the vector of $q$ independent variables, $n_u$ and $n_y$ are the maximum lags of the input and output signals, respectively.

For the sake of clarity, only one exogenous input will be considered and no input delay is considered although the procedure can easily be extended to multi-input systems with delays. See \cite{haf_eal/19} for a comprehensive investigation on the use of multi-objetive techniques in nonlinear system identification.

\subsection{Measured data}

It is assumed that a set of dynamical data is available to estimate $\boldsymbol{\theta}$, organized as 
\begin{equation}
\boldsymbol{Z}_{\rm d}=[\boldsymbol{\psi}(k-1)\,\,\,y(k)], \nonumber
\end{equation}

\noindent 
where $k=1,\dots,N_{\rm d}:\,\boldsymbol{Z}_{\rm d}\in\mathbb{R}^{N_{\rm d}\times(n_y+n_u+2)}$. Another dynamical dataset, named $\boldsymbol{Z}_{\rm t}$ (test data), with the same characteristics of $\boldsymbol{Z}_{\rm d}$, is also considered. Ordinary \emph{black box} approaches use only dynamical datasets ($\boldsymbol{Z}_{\rm d}$ and $\boldsymbol{Z}_{\rm t}$) as measured source of information.

The auxiliary information about the \emph{steady-state} behavior of the system is expressed as a set of pairs $(\bar{u}_j,\,\bar{y}_j),~j=1,\dots,N_{\rm s}$. Hence $(\bar{u}_1,\,\bar{y}_1)$ says that if the input $u(k)=\bar{u}_1$ is held constant for a sufficiently long time, the output $\lim_{k\to\infty} y(k)=\bar{y}_1$. The steady-state information can be represented by
\begin{equation}
\boldsymbol{Z}_{\rm s}=[{\bar u}_j\,\,\,\bar y_j], \nonumber
\end{equation}

\noindent where $j=1,\dots,N_{\rm s}:\,\boldsymbol{Z}_{\rm s}\in\mathbb{R}^{N_{\rm s}\times2}$.

\subsection{Problem Statement}

For given model structure of the class shown in (\ref{eq:model}) and data sets $\boldsymbol{Z}_{\rm d}$ and $\boldsymbol{Z}_{\rm s}$, the aim is to  estimate $\boldsymbol{\theta}$ in such a way as to \textit{simultaneously} minimize error functions on $\boldsymbol{Z}_{\rm d}$ {\it and}\,~$\boldsymbol{Z}_{\rm s}$.

\section{Background}
\label{sec:backGr}

\subsection{MLP networks}

Neural network models are often implemented due to their properties as universal approximators, which can exhibit good performance in the context of dynamical systems \cite{chen1990,barbosa2011,Asteris2017}. The main challenge in the training of such structures is the fact that they are nonlinear-in-the-parameters. This leads to nonconvex estimation problems, for which backpropagation (BP) is one of the most used algorithms.

This paper considers MLP networks of the form
\begin{eqnarray} 
\label{eq:mlp}
  y(k) &=& F\left( \boldsymbol{u},\,\boldsymbol{y},\,\boldsymbol{\theta} \right) \nonumber\\
       &=& \theta_0 + \sum_{i=1}^{n_h} \theta_i \tanh \bigg( \theta_{i,0} + \sum_{j=1}^{n_y} \theta_{i,j} y(k-j) + \sum_{j=1}^{n_u} \theta_{i,(j+n_y)} u(k-j) \bigg),
\end{eqnarray}
where $n_h$ denotes the number of neurons in the hidden layer, a structural parameter assumed to be known. Vectors $\boldsymbol{u}\in \mathbb{R}^{n_u}$ and $\boldsymbol{y}\in \mathbb{R}^{n_y}$ indicate all
the lagged values of the input $u(k)$ and output $y(k)$ -- and eventually a constant -- used in the model.

Using the standard BP algorithm it is possible to fit model (\ref{eq:mlp}) to the dynamical data $\boldsymbol{Z}_{\rm d}$, but in that case the static data set $\boldsymbol{Z}_{\rm s}$, which is the auxiliary information, would not be used during  training. 

Steady-state analysis of model (\ref{eq:mlp}) is carried out by taking a constant input $u(k-j)=\bar{u},\,\forall j=1,\,\dots , n_u $ and applying it to the model. We here assume that the trained model with parameter vector $\hat{\boldsymbol{\theta}}$ is asymptotically stable, hence the output will converge to the output $\hat{\bar{y}}$ which is the solution to the following algebraic equation
\begin{eqnarray} 
\label{eq:mlp_ss}
\hat{\bar{y}} &=& \bar{F}\left( \bar{u},\,\hat{\bar{y}},\,\hat{ \boldsymbol{\theta} }\right) \nonumber\\
       0 &=& \bar{F}\left( \bar{u},\,\hat{\bar{y}},\,\hat{ \boldsymbol{\theta} }\right) - \hat{\bar{y}},
\end{eqnarray}
where $\bar{F}$ is given by
\begin{eqnarray} 
\label{eq:mlp_2}
  \hat{\bar{y}} &=& \bar{F}\left( \bar{u},\,\hat{\bar{y}},\,\hat{\boldsymbol{\theta}} \right) \nonumber\\
       &=& \hat{\theta}_0 + \sum_{i=1}^{n_h} \hat{\theta}_i \tanh \bigg( \hat{\theta}_{i,0} + \hat{\bar{y}}\sum_{j=1}^{n_y} \hat{\theta}_{i,j} + \bar{u}\sum_{j=1}^{n_u} \hat{\theta}_{i,(j+n_y)} \bigg).
\end{eqnarray}

The values of
$\hat{\bar{y}}$ that satisfy (\ref{eq:mlp_ss}) for the chosen $\bar{u}$ are called the \emph{fixed points} of $F$ for the given $\bar{u}$. The number of fixed points can vary from none to several and such fixed points can be asymptotically stable, stable or unstable. Many real processes have one asymptotically stable fixed point. This means that for the input $\bar{u}$ the process output will asymptotically converge to $\hat{\bar{y}}$ if the initial conditions are within the basin of attraction of $\hat{\bar{y}}$. After reaching steady-state, the system will remain at $y(k)=\hat{\bar{y}},~\forall k$ until the input is changed or the system is perturbed in any other way.

The fixed point $\hat{\bar{y}}$ is very challenging to calculate analytically for the model structure~\eqref{eq:mlp_2}, because it may not be possible to obtain an equation of the form $\hat{\bar{y}}=\bar{F}_s\left( \bar{u},\,\hat{\boldsymbol{\theta}} \right)$, where $\bar{F}_s(\cdot)$ does not depend on $\hat{\bar{y}}$. In order to use the steady-state information, $\boldsymbol{Z}_{\rm s}$, the work~\cite{amaral2001} used a simpler model structure, proposed in~\cite{narendra1990}, where the autorregressive terms are linear-in-the-parameter, of the form $y(k) = F\left( \boldsymbol{u},\,\boldsymbol{\theta}_u \right) + \boldsymbol{\theta}_y \boldsymbol{y}$, where it is always possible to write $\hat{\bar{y}}=\bar{F}_s\left( \bar{u},\,\boldsymbol{\theta}_u,\,\boldsymbol{\theta}_y \right)$.
No such simplification is required by the method proposed in the present paper.

For polynomial models, that are linear-in-the-parameter, a bi-objective approach was proposed in \cite{barbosa2011}, where the vector of parameters $\hat{\boldsymbol{\theta}}$ was estimated by simultaneously minimizing the functions
\begin{eqnarray} 
  J_{\rm d} &=& \frac{1}{N_{\rm d}} \sum_{k=1}^{N_{\rm d}} [ y(k)-\hat{y}(k) ]^2, \label{eq:cost_abreu2012}\\
  J_{\rm s} &=& \frac{1}{N_{\rm s}} \sum_{j=1}^{N_{\rm s}} [ \bar{y}_j-\hat{\bar{y}}_j ]^2, \label{eq:cost_abreu2012s}
\end{eqnarray}
where $\hat{y}(k)$ corresponds to the free-run simulation, the $\hat{\bar{y}}_j$ values were found analytically, $J_{\rm d}$ and $J_{\rm s}$ are computed over $\boldsymbol{Z}_{\rm d}$ and $\boldsymbol{Z}_{\rm s}$ datasets, respectively. An evolutionary algorithm was used for parameters estimation due to the nonconvexity of the optimization problem due to the fact that free-run simulated values are used in $J_{\rm d}$ \cite{barbosa2011}. One of the aims of the present paper is to remove the necessity of analytically obtaining $\hat{\bar{y}}_j$ prior to the optimization step, as discussed below.

The main drawback of such an approach for training MLP networks~\eqref{eq:mlp_2} is the calculation of $\hat{\bar{y}}$. To see this, consider the {\it measured}\, static point $(\bar{u},\,\bar{y})$. The model fixed point $\hat{\bar{y}}$ must be calculated by solving (\ref{eq:mlp_ss}). This is usually costly as there is no general analytical solution. Alternatively, $\hat{\bar{y}}$ can be obtained numerically by recursively iterating the dynamical model $F\left( \boldsymbol{u},\,\boldsymbol{y},\,\hat{\boldsymbol{\theta}} \right)$ in (\ref{eq:mlp}) with
\begin{eqnarray} 
\boldsymbol{u}=[\bar{u} ~\bar{u} \ldots \bar{u}]^T \in \mathbb{R}^{n_u} \nonumber
\end{eqnarray}
until convergence. The value to which the model converges is $\hat{\bar{y}}$. Because this has to be accomplished at each training step and for each measured static point, the procedure is computationally costly. This procedure is used in the simulated example in Sec.~\ref{sec:results_ex3}, to illustrate the main benefit of the method proposed in this paper: the lower computational cost with good performance.

In the case of NARX polynomial models, the fixed points are clearly related to term clusters and cluster coefficients \cite{aguirre2004}. For this model class, constrained and bi-objective optimization algorithms can be readily used to take advantage of the auxiliary information about the location of fixed points and static nonlinearities \cite{nepomuceno2007,barbosa2011,Aguirre_2019}. Unfortunately such procedures cannot be easily extended to more complex model structures.



The method presented in the next section aims at circumventing the shortcomings pointed out in the two last paragraphs.

\section{Proposed Methodology}
\label{sec:methods}




Many black-box procedures estimate the parameters by minimizing the cost function (\ref{eq:cost_abreu2012}),  with $\hat{y}(k)$ being the \emph{one step ahead} prediction instead of the free-run simulation, over the dynamical (training) data set $\boldsymbol{Z}_{\rm d}$. This choice, although very convenient from a numerical point of view, does not take full advantage of free-run simulations \cite{piroddi2003,Piroddi2008,rib_agu/18}.
One possible solution to use information of the static data set $\boldsymbol{Z}_{\rm s}$ during  parameter estimation is implementing a bi-objective optimization problem where another cost function, say (\ref{eq:cost_abreu2012s}), is simultaneously considered. Although these objective functions may not be considered as ``conflicting'', since the system itself provided both data sets, it is quite hard to have dynamic and static information equally weighted in a single data set when dealing with nonlinear systems. Also, because steady-state historical data can be readily averaged, it is easier to have $\boldsymbol{Z}_{\rm s}$ of better quality than $\boldsymbol{Z}_{\rm d}$.

Computing (\ref{eq:cost_abreu2012s}) requires finding the fixed points $\hat{\bar{y}}_j$ of the model, which is in general  computationally expensive. Thus, the key feature in the proposed methodology is that the fixed points do not need to be explicitly computed neither analytically (e.g. for polynomial models) nor numerically (e.g. for ANN). Instead, here it is proposed to minimize:
\begin{equation}
\label{eq:cost_static}
  \hat{J}_{\rm s} = \frac{1}{N_{\rm s}} \sum_{j=1}^{N_{\rm s}} \left[ \bar{y}_j-F\left( \boldsymbol{\bar\psi}_j,\,\hat{\boldsymbol{\theta}} \right)\right]^2,
\end{equation}
where the hat over ${J}_{\rm s}$ indicates that (\ref{eq:cost_static}) is an approximation to (\ref{eq:cost_abreu2012s}). Here, as before, $\bar{y}_j$ can be seen as a ``target value'' taken from the static data and 
\begin{equation}
  \label{eq:psi_bar}
  \boldsymbol{\bar\psi}_j=\left[1 \,\,\,\,\,\,  \bar{y}_j\dots \bar{y}_j\,\,\,\,\,\, \bar{u}_j \dots  \bar{u}_j \right]^T \in \mathbb{R}^{1+n_y+n_u} .
\end{equation}

\noindent
It should be noted that $F\left( \boldsymbol{\bar\psi}_j,\,\hat{\boldsymbol{\theta}} \right)$ is simply the model one-step-ahead prediction.

\begin{lem}
  \label{lema}
  Both (\ref{eq:cost_abreu2012s}) and (\ref{eq:cost_static}),  computed over $\boldsymbol{Z}_{\rm s}$, for each corresponding input $\bar{u}_j$, have global minima $J_{\rm s}=\hat{J}_{\rm s}=0$ at the model fixed points $\bar{y}_j$, $j=1,\dots , N_s$.
\end{lem}
\begin{proof}
  If $\bar{y}_j, \forall j=1,\dots , N_s$ are fixed points of $F$, then $\bar{y}_j=\hat{\bar{y}}_j$ and therefore
  from (\ref{eq:cost_abreu2012s}) it follows immediately that $J_{\rm s}=0$.
  Now, if $F$ is initialized at $\bar{y}_j$ and $\bar{u}_j$ by
  taking $\boldsymbol{\bar\psi}_j$ from (\ref{eq:psi_bar}), since
  the vector field is null at that point, the one step ahead prediction will necessarily be
  $F( \boldsymbol{\bar\psi}_j,\,\hat{\boldsymbol{\theta}} )=\bar{y}_j, \forall j=1,\dots , N_s$. Hence, at the fixed points $\hat{J}_{\rm s}=0$.
\end{proof}

Hence, although the model fixed points are not explicitly used in (\ref{eq:cost_static}) as in Eq.\,\ref{eq:cost_abreu2012s},
both $J_{\rm s}$ and $\hat{J}_{\rm s}$ reach minima at fixed points. While $J_{\rm s}$ only 
uses static data (measured and from the model), $\hat{J}_{\rm s}$ uses both: the target $\bar{y}_j$ which is a fixed value and the model output $F( \boldsymbol{\bar\psi}_j,\,\hat{\boldsymbol{\theta}} )$ which is obtained by performing {\it one}\, iteration of the model using the target $\boldsymbol{\bar\psi}_j$ as initial condition, as shown in Figure~\ref{fig:method}. During training $\bar{y}_j$ and $\bar{u}_j$ might not yet be exactly a fixed point of the model and the one step ahead prediction will be somewhat different from the target. Hence
after one iteration, Eq.\,\ref{eq:cost_static} is used to evaluate how far did the model move away from the target. Therefore, if the model parameters are estimated by minimizing $\hat{J}_{\rm s}$, this will result in models with equilibria close to $\bar{y}_j$. Based on Lemma~\ref{lema}, the following methodology is proposed.

\begin{figure}
 \centering
  \includegraphics[width=0.95\textwidth]{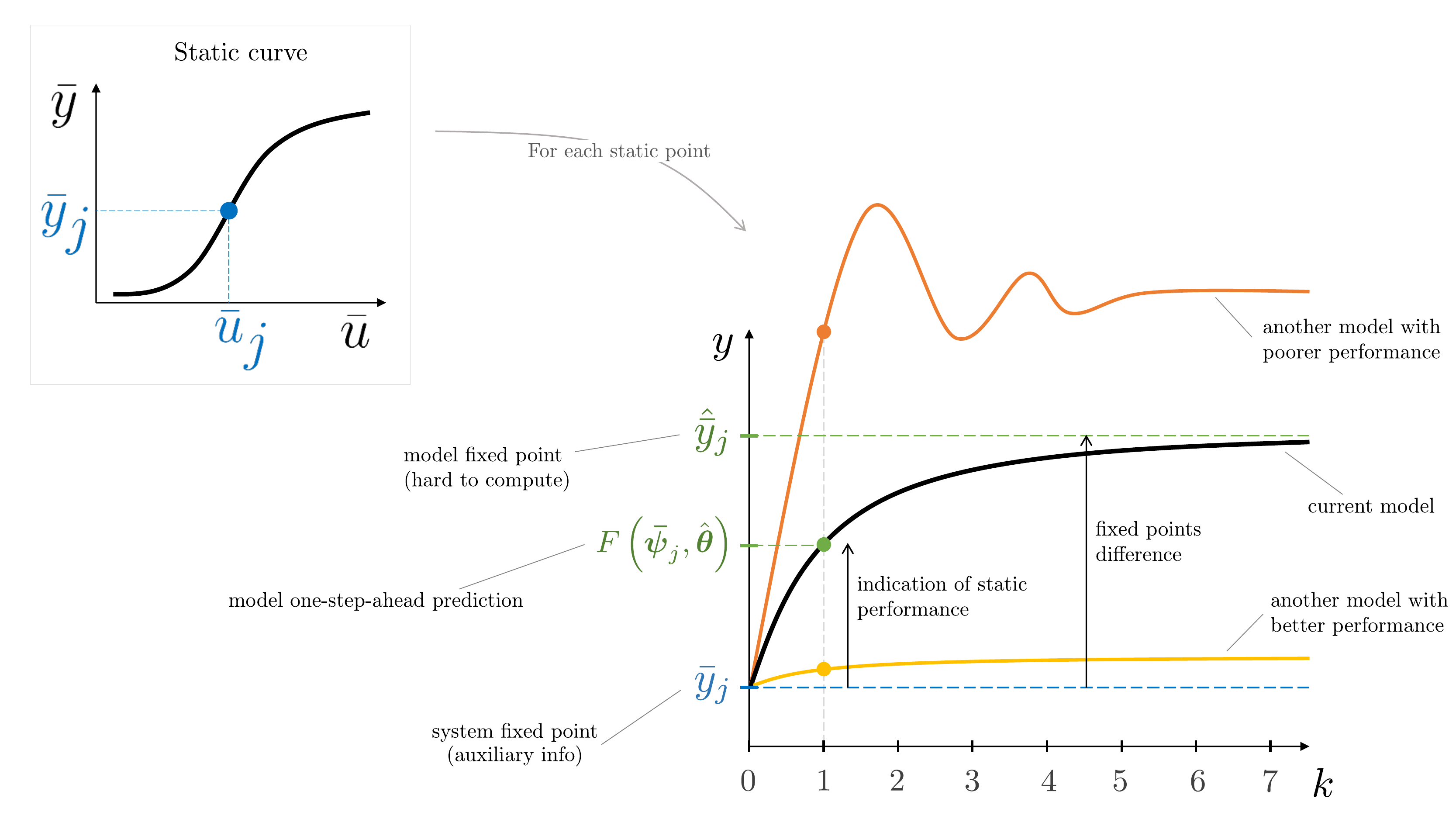}
  \caption{Free-run simulation, from the same initial condition $\bar{y}_j$, of three hypothetical models with different static performance. The one-step-ahead predictions, indicated by dots, are a low-cost indication of static performance -- while the estimation of model fixed points are costly. The models with better (yellow) and poorer (red) static performance are properly quantified by the one-step-ahead prediction, used to train the model.
  }
 \label{fig:method}
\end{figure}


\begin{propos}
  \label{proposition}
  A way of using both dynamic and static information in model building is by minimizing the following cost function, that is a convex combination of $\hat{J}_{\rm s}$ and $J_{\rm d}$:
  \begin{equation}
  \label{eq:cost_sd}
    \boldsymbol{J_{\rm sd}} = (1-\lambda) J_{\rm d} + \lambda \hat{J}_{\rm s},
  \end{equation}
  where $\lambda\in [0,\,1]$ is the parameter that weights the balance between \emph{static} and \emph{dynamical} information.
\end{propos}

The use of (\ref{eq:cost_static}) instead of (\ref{eq:cost_abreu2012s}) in (\ref{eq:cost_sd}) is the key-point of the proposed method. This change reduces the computational cost by approximately three orders of magnitude while keeping the model performance competitive.

When $\lambda=0$ the estimation algorithm only considers the dynamical information (e.g. \emph{black-box} approach) and as $\lambda\to 1$ the influence of the auxiliary information about the system in steady-state gradually increases, as shown in Figure~\ref{fig:pareto}. Generally, competitive models can be achieved by a suitable balance of the information in $\boldsymbol{Z}_{\rm s}$ and $\boldsymbol{Z}_{\rm d}$ datasets \cite{barbosa2011,barroso2007}. 


Finding an adequate value for $\lambda$ is carried out in the examples using two decision makers. The first, proposed in \cite{barroso2007}, considers the free-run simulation error over $\boldsymbol{Z}_{\rm d}$, in which the model with the minimum correlated error with the system output is chosen. The second measures the root mean squared error (RMSE) of the free-run simulation over $\boldsymbol{Z}_{\rm t}$ and chooses the smallest. 

After choosing the value of $\lambda$, the minimization of (\ref{eq:cost_sd}) can be solved by standard  algorithms, where the choice often depends on the model structure, as detailed next.


\begin{figure}
 \centering
  \includegraphics[width=0.70\textwidth]{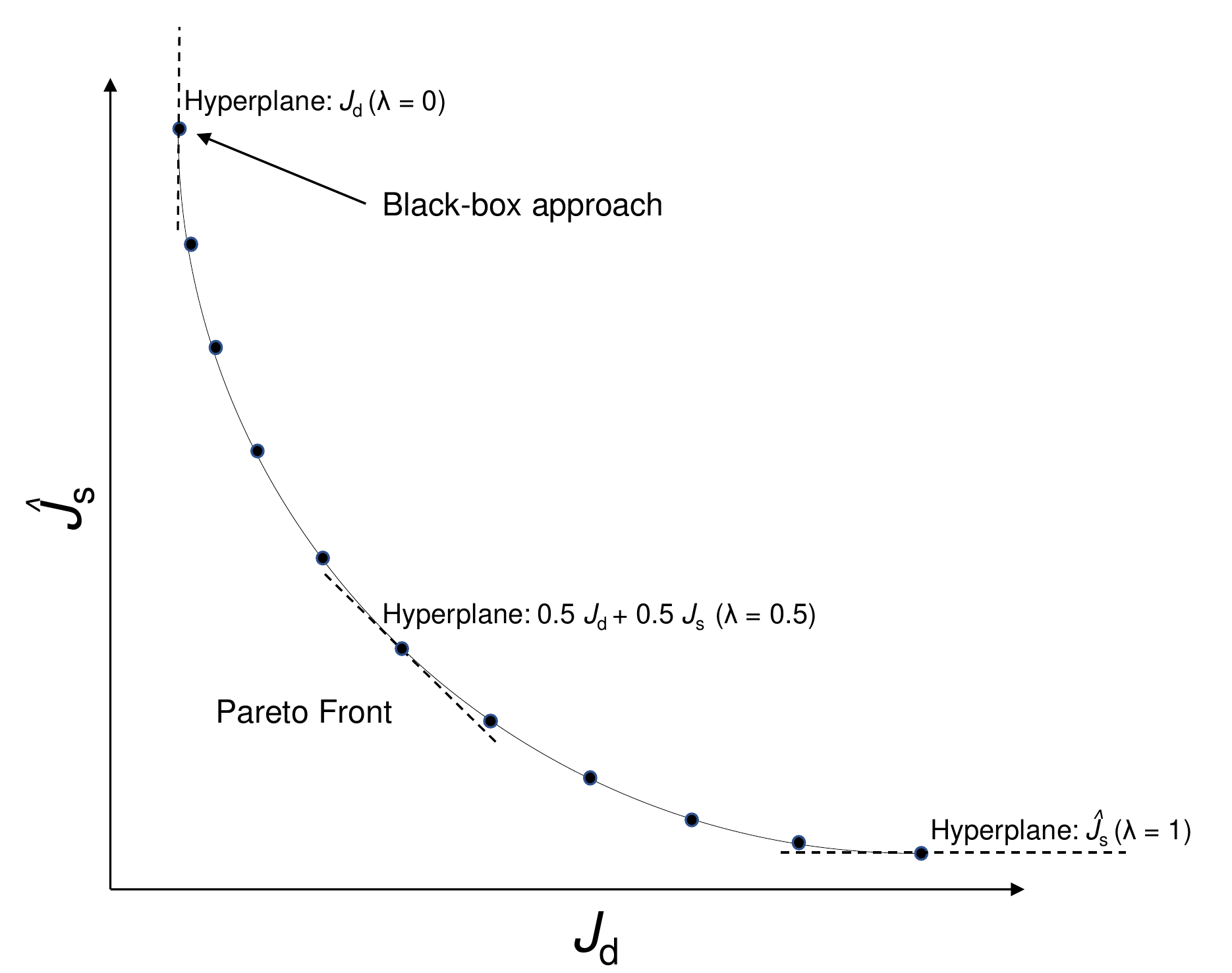}
  \caption{Pareto Front obtained by minimizing (\ref{eq:cost_sd}) using different values for $\lambda$. Each dot corresponds to a model with a different balance between the used cost functions.}
 \label{fig:pareto}
\end{figure}

\subsection{Linear-in-the-parameter models}

For linear-in-the-parameter models, the weighted least squares (WLS) can be used with:
\begin{equation}
\label{eq:weightMatrix_WLS}
  W =
  \begin{bmatrix}
    (1-\lambda)\boldsymbol{I}_{N_{\rm d}} & \boldsymbol{0} \\
    \boldsymbol{0} & \lambda\boldsymbol{I}_{N_{\rm s}}
  \end{bmatrix},\,
  \boldsymbol{Y} =
  \begin{bmatrix}
    \boldsymbol{y} \\ \boldsymbol{\bar{y}}
  \end{bmatrix},\,
  \boldsymbol{\Psi} =
  \begin{bmatrix}
    \boldsymbol{\psi} \\ \boldsymbol{\bar{\psi}}
  \end{bmatrix},\,
\end{equation}
where $\boldsymbol{I}_{N}\in\mathbb{R}^{N\times N}$ is the identity matrix, $\boldsymbol{0}$ is the null matrix of appropriate dimension and the pairs $(\boldsymbol{\psi},\boldsymbol{y})$ and $(\boldsymbol{\bar{\psi}},\boldsymbol{\bar{y}})$ are available in $\boldsymbol{Z}_{\rm d}$ and $\boldsymbol{Z}_{\rm s}$, respectively. The solution is given by $\hat{\boldsymbol{\theta}}=(\boldsymbol{\Psi}^\top W\,\boldsymbol{\Psi})^{-1} \boldsymbol{\Psi}^\top W\,\boldsymbol{Y} $.

It is possible to prove that the above solution for polynomial models is rigorously equivalent to that presented by~\cite{nepomuceno2007}. This is shown numerically in Sec.~\ref{sec:piroddi}.

\subsection{Nonlinear-in-the-parameter models}

Fortunately, unlike in Ref.~\cite{nepomuceno2007}, the proposed methodology can be also used to estimate parameters of models that are nonlinear-in-the-parameters. This can be done using the following error vector:
\begin{equation}
\label{eq:MLP_errorTerm}
  \boldsymbol{E} =
  \begin{bmatrix}
    \,\,\,(1-\lambda)\left[\boldsymbol{y} - \boldsymbol{F} \left( \boldsymbol{\psi},\,\boldsymbol{\hat\theta} \right)\right] \,\,\, \\
    \,\,\,\lambda\left[\boldsymbol{\bar y} - \boldsymbol{F} \left( \boldsymbol{\bar\psi},\,\hat{\boldsymbol{\theta}} \right)\right] \,\,\,
  \end{bmatrix} \quad \in \quad \mathbb{R}^{(N_{\rm d}+N_{\rm s})},
\end{equation}
where $\boldsymbol{F}(\cdot)$ is the vector of the model outputs, $\boldsymbol{F}(\boldsymbol{\psi},\,\boldsymbol{\hat\theta})=[F( \boldsymbol{\psi}(0),\,\hat{\boldsymbol{\theta}}) ~F( \boldsymbol{\psi}(1),\,\hat{\boldsymbol{\theta}}) \cdots F( \boldsymbol{\psi}(N_d-1),\,\hat{\boldsymbol{\theta}})]^T$ and $\boldsymbol{F}(\boldsymbol{\bar\psi},\,\boldsymbol{\hat\theta})=[F( \boldsymbol{\bar\psi}_1,\,\hat{\boldsymbol{\theta}}) ~F( \boldsymbol{\bar\psi}_2,\,\hat{\boldsymbol{\theta}}) \cdots F( \boldsymbol{\bar\psi}_{N_s},\,\hat{\boldsymbol{\theta}})]^T$.


The error vector~\eqref{eq:MLP_errorTerm} allows the parameter estimation by minimizing $\boldsymbol{J_{\rm sd}}$ (Proposition~\ref{proposition}) using standard algorithms, like weighted BP and Levenberg-Marquardt -- used in this paper. The main advantage of~\eqref{eq:MLP_errorTerm} is that the static part is much easier to compute than previous methods that estimate the fixed point of the \emph{model}. To compute $\boldsymbol{E}$, the model $F(\cdot)$ must be iterated  $(N_d+N_s)$ times. In previous procedures, to estimate a single fixed point $\hat{\bar{y}}_j$ the model had to be iterated many times, e.g. $F(F(F(F(\dots F(\cdot)))))$ until steady-state was reached. If the model required, on average, $k_{ss}$ iterations to reach  steady-state, the computation of the error vector by previous methods would require $(N_d+N_s \times k_{ss})$ model iterations, compared to $(N_d+N_s)$ in the proposed methodology.

The main steps of the proposed methodology are summarized below:
\begin{enumerate}
    \item Begin with a given model structure $\boldsymbol{F}(\cdot)$ and initial parameter vector  $\hat{\boldsymbol{\theta}}_0$;
    \item With $\boldsymbol{Z}_{\rm d}$, compute $\left[\boldsymbol{y} - \boldsymbol{F} \left( \boldsymbol{\psi},\,\hat{\boldsymbol{\theta}}_0 \right)\right]$;
    \item With $\boldsymbol{Z}_{\rm s}$, compute the vector of independent variables $\boldsymbol{\bar\psi}$ \eqref{eq:psi_bar} and $\left[\boldsymbol{\bar y} - \boldsymbol{F} \left( \boldsymbol{\bar\psi},\,\hat{\boldsymbol{\theta}}_0 \right)\right]$;
    \item Compute the error vector $\boldsymbol{E}$ as in~\eqref{eq:MLP_errorTerm} for values within $\lambda\in [0,1]$;
    \item Run an optimization algorithm (e.g. Levenberg-Marquardt) to minimize $\boldsymbol{E}$ and obtain $\hat{\boldsymbol{\theta}}(\lambda)$;
    \item Use some criterion to choose $\lambda$ that gives the best $\hat{\boldsymbol{\theta}}$ (decision making).
\end{enumerate}
In what follows,
the above procedure is illustrated in numerical examples and in the estimation of downhole pressure soft-sensors.




\section{Results and Discussion}
\label{sec:results}

The proposed methodology is now applied to two simulated systems (taken from \cite{piroddi2003} and \cite{jakubek2008}) and to a real deep water oil well process \cite{aguirre2017cep}. Following the main aim of the paper, all cases require a grey-box approach to achieve suitable results. The proposed methodology is compared with other methods in terms of accuracy and computational cost.

The first example shows that, for polynomial models, the results are equivalent to those using \cite{nepomuceno2007}. The second example shows that, with MLP models, the computational cost of using the proposed methodology is much less than employing the available grey-box approaches that can be applied to such models \cite{barbosa2011}. The last result shows that the method can attain competitive results on a real deep water oil well process.

\subsection{Simulated Example 1}
\label{sec:piroddi}

Consider the dynamical nonlinear system \cite{piroddi2003}:
\begin{eqnarray} 
\label{eq:ex1_sys}
  w(k) &=& 0.75 w(k-2) + 0.25 u(k-1) - 0.2 w(k-2) u(k-1),\nonumber \\
  y(k) &=& w(k) + e(k),
\end{eqnarray}
where $u \in \mathbb{R}$ is the input, $w \in \mathbb{R}$ the noiseless output, $y \in \mathbb{R}$ the output with the noise $e(k)~\sim~\rm{WGN}$(0,\,0.1${\sigma_w}$), where $\rm{WGN}$ stands for the White Gaussian Noise.

Four datasets were obtained from system (\ref{eq:ex1_sys}):   $\boldsymbol{Z}_{\rm d}$, $\boldsymbol{Z}_{\rm t}$ and  $\boldsymbol{Z}_{\rm s}$) were used in parameter estimation; $\boldsymbol{Z}_{\rm v}$ was used to compare estimation techniques. In order to represent a common situation in practice, the training and testing datasets, $\boldsymbol{Z}_{\rm d}$ and $\boldsymbol{Z}_{\rm t}$, respectively, were acquired over a limited operating range, however, with a persistently exciting input. The steady-state dataset $\boldsymbol{Z}_{\rm s}$ was obtained over a wider operating range.

\begin{rem}
  The validation dataset $\boldsymbol{Z}_{\rm v}$ was simulated over a wider operating range to allow a more thorough comparison. In many practical problems this dataset is not available.
\end{rem}

The dynamical training dataset $\boldsymbol{Z}_{\rm d}$ and testing dataset $\boldsymbol{Z}_{\rm t}$ (used only for decision-making process) were simulated with $u\sim\rm{WGN}(-0.02,\,0.04)$, $N_{\rm d}=100$, $N_{\rm t}=400$ and $ e \sim \rm{WGN}$(0,\,0.1$\sigma_w$) . The static dataset $\boldsymbol{Z}_{\rm s}$ was obtained analytically with $N_{\rm s}=50$ equally spaced values in the range $u\in [-1,\,3]$ and with an additive zero mean noise with  $\sigma=0.02$. The validation dataset $\boldsymbol{Z}_{\rm v}$ was simulated over a broader operating range, with $N_{\rm v} = 2000$ samples and without output noise ($e=0$). Note that $\boldsymbol{Z}_{\rm v}$ and $\boldsymbol{Z}_{\rm s}$ have inputs with wide spectral range, but $\boldsymbol{Z}_{\rm d}$ and $\boldsymbol{Z}_{\rm t}$ lack  information in operating ranges far from $y \approx 0$.

Figure~\ref{fig:ex1_data_piroddi_zdzszv} compares the analytic static curve with all datasets in the $(u,y)$ plane. Clearly, information in $(\boldsymbol{Z}_{\rm d},\boldsymbol{Z}_{\rm t})$ and $\boldsymbol{Z}_{\rm s}$ are complementary.

\begin{figure}
 \centering
  \includegraphics[width=1.0\textwidth]{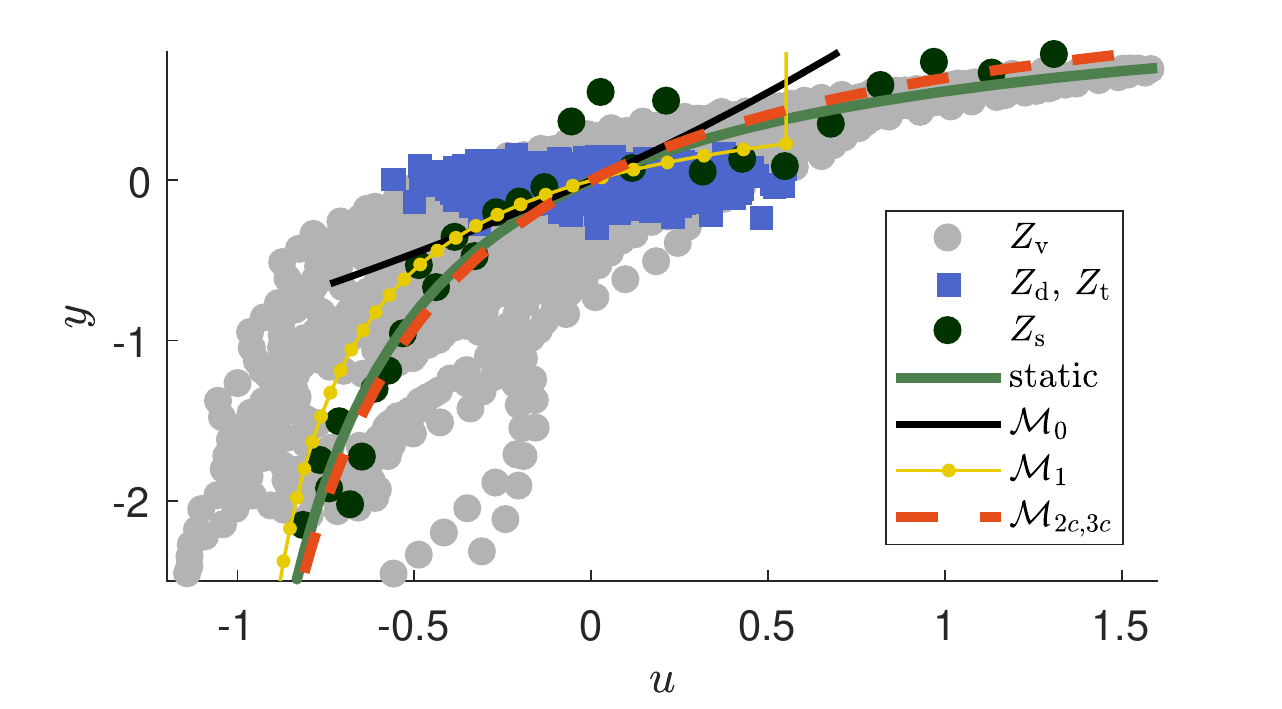}
  \caption{Comparison of the analytical static curve (thick green line) of system (\ref{eq:ex1_sys}) with datasets ($\boldsymbol{Z}_{\rm v}$, $\boldsymbol{Z}_{\rm d}$, $\boldsymbol{Z}_{\rm t}$, $\boldsymbol{Z}_{\rm s}$) and with the estimated static behavior of the models ($\mathcal{M}_0$, $\mathcal{M}_1$, $\mathcal{M}_{2c}$ and $\mathcal{M}_{3c}$). The static behavior of $\mathcal{M}_{2c}$ and $\mathcal{M}_{3c}$ were equal, with precision of $10^{-10}$.}
 \label{fig:ex1_data_piroddi_zdzszv}
\end{figure}


The following polynomial NARX structure was obtained from the data produced by system~\eqref{eq:ex1_sys} 
\begin{align} \label{eq:ex1_m0}
  y(k) & = \theta_1 y(k-2) + \theta_2 u(k-1) + \theta_3 u(k-1)y(k-2) \nonumber \\
       & + \theta_4 u(k-1)y(k-1) + \theta_5 u(k-2)y(k-1) ,
\end{align}
using the procedure proposed in \cite{mendes2001}. Four parameter estimation techniques were tested,
each one yielding a different model family: $\mathcal{M}_1$ obtained by the Constrained Least Squares (CLS) as in \cite{aguirre2004}, $\mathcal{M}_2$ obtained as in~\cite{nepomuceno2007} and $\mathcal{M}_3$ found using  Weighted Least Squares (WLS) as described in Sec.~\ref{sec:methods}. $\mathcal{M}_0$ was the model obtained following a black-box approach with ordinary Least Squares (LS), shown to illustrate the disadvantages of not using auxiliary information in this example.

For the sake of comparison, nine models of each $\mathcal{M}_2$ and $\mathcal{M}_3$ were estimated, with $\lambda\in\{0.1,\,0.2,\,0.3,\,0.4,\,0.5,\,0.6,\,0.7,\,0.8,\,0.9\}$, and  two decision makers were adopted: minimum correlation \cite{barroso2007} ($\mathcal{M}_{2a}$,~$\mathcal{M}_{3a}$); and minimum free-run RMSE over $\boldsymbol{Z}_{\rm t}$ ($\mathcal{M}_{2b}$,~$\mathcal{M}_{3b}$). In addition, the best ($\mathcal{M}_{2c}$, $\mathcal{M}_{3c}$) and worst ($\mathcal{M}_{2d}$,~$\mathcal{M}_{3d}$) models in terms of RMSE over $\boldsymbol{Z}_{\rm v}$ are shown. Table~\ref{tab:models} summarizes the parameter estimation techniques used to identify the three studied examples.

\begin{table}
  \caption{Parameter estimation techniques per model family.}
  \label{tab:models}
  \centering{
  \begin{tabular}[t]{l|ccccc}
    & \bf{$\mathcal{M}_0$} & \bf{$\mathcal{M}_1$} & \bf{$\mathcal{M}_2$} & \bf{$\mathcal{M}_3$} \\ \hline
    {\bf Ex.1} \cite{piroddi2003} & LS & CLS \cite{aguirre2004} & \cite{nepomuceno2007} & WLS (Sec.~\ref{sec:methods}) \\
    {\bf Ex.2} \cite{jakubek2008} & BP & -- & evolutionary \cite{barbosa2011} & BP (Sec.~\ref{sec:methods}) \\
    {\bf Ex.3} \cite{aguirre2017cep} & BP & -- & evolutionary \cite{barbosa2011} & BP (Sec.~\ref{sec:methods}) \\
    {aux.info} & -- & constraint & multi-obj. & multi-obj. \\ \hline
  \end{tabular}}
\end{table}

Part of the free-run simulation over validation dataset $\boldsymbol{Z}_{\rm v}$ is shown in Figure~\ref{fig:ex1_data_piroddi_results}. The black-box model ($\mathcal{M}_0$) shows that $(\boldsymbol{Z}_{\rm d},\boldsymbol{Z}_{\rm t})$ are not sufficient to achieve a good performance over a wide operating range. $\mathcal{M}_1$ is not shown because it becomes unstable over $\boldsymbol{Z}_{\rm v}$, as a consequence of imposing inaccurate auxiliary information via hard constraints.  

The bi-objective estimation, with $\mathcal{M}_2$ and $\mathcal{M}_3$, allowed a better trade-off between auxiliary information and the dynamical data, as shown in Figure~\ref{fig:ex1_data_piroddi_zdzszv} with the best models $\mathcal{M}_{2c}$ and $\mathcal{M}_{3c}$ simulations. The root mean squared error (RMSE) over free-run simulation on the validation dataset is summarized in Table~\ref{tab:ex_results_rmse}. The performance of both estimation techniques ($\mathcal{M}_2$, $\mathcal{M}_3$) were rigorously equivalent as well as the computational cost. It is worth to mention that, using auxiliary information, the proposed approach found a model with good performance over a much broader operating range dataset ($\boldsymbol{Z}_{\rm v}$), when compared with the dataset used for  parameters estimation $(\boldsymbol{Z}_{\rm d})$,  without computing fixed points. Besides, the identified models $\mathcal{M}_{3b}$ and $\mathcal{M}_{3c}$ achieved good static and dynamic performance.

\begin{figure}
 \centering
  \includegraphics[width=1.0\textwidth]{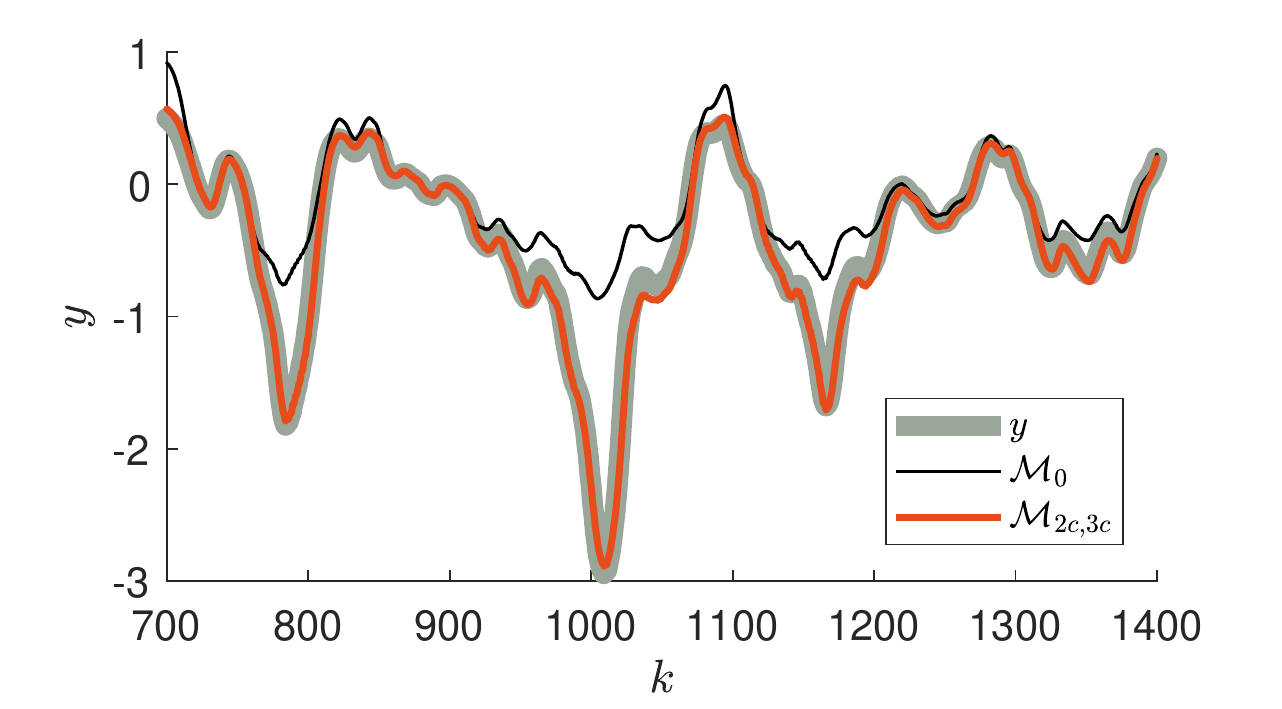}
  \caption{Free-run simulation over validation dataset $\boldsymbol{Z}_{\rm v}$, where system output $y$ is the bold grey line; $\mathcal{M}_0$ (\emph{black-box} NARX poynomial model) the black line;  $\mathcal{M}_{2c}$ (\emph{grey-box} polynomial model) and $\mathcal{M}_{3c}$ (\emph{grey-box} MLP model) the red line.
  }
 \label{fig:ex1_data_piroddi_results}
\end{figure}

\begin{table}
  \caption{Root mean squared error (RMSE) of each model evaluated over validation dataset $\boldsymbol{Z}_{\rm v}$ in a free-run simulation.}
  \label{tab:ex_results_rmse}
  \centering\scriptsize{
  \begin{tabular}[t]{l|cccl}
    & {\bf Ex.1 \cite{piroddi2003}} & {\bf Ex.2 \cite{jakubek2008}} & {\bf Ex.3 \cite{aguirre2017cep}} & {\bf Choice of $\lambda$} \\ \hline\hline
    {\bf$\mathcal{M}_{0}$} & $0.4649$ & $0.3182$ & $6.7420$ & -- \\ \hline 
    {\bf$\mathcal{M}_{1}$} & $\gg 10^2$ & -- & -- & -- \\ \hline
    {\bf$\mathcal{M}_{2a}$} & $38.349$ ($\lambda=0.8$) & $0.2211$ ($\lambda=0.7$) & -- & {min corr. \cite{barroso2007}} \\
    {\bf$\mathcal{M}_{2b}$} & $0.0557$ ($\lambda=0.1$) & $0.2211$ ($\lambda=0.7$) & -- & {min RMSE over $\boldsymbol{Z}_{\rm t}$} \\
    {\bf$\mathcal{M}_{2c}$} & $0.0557$ ($\lambda=0.1$) & $0.1190$ ($\lambda=0.9$) & -- & {min RMSE over $\boldsymbol{Z}_{\rm v}$} \\
    {\bf$\mathcal{M}_{2d}$} & $\gg 10^2$ ($\lambda=0.9$) & $0.2899$ ($\lambda=0.1$) & -- & {max RMSE over $\boldsymbol{Z}_{\rm v}$} \\ \hline
    {\bf$\mathcal{M}_{3a}$} & $38.349$ ($\lambda=0.7$) & $0.0992$ ($\lambda=0.1$) & $23.9496$ $(\lambda=0.68)$ & {min corr. \cite{barroso2007}} \\
    {\bf$\mathcal{M}_{3b}$} & $0.0557$ ($\lambda=0.1$) & $0.0992$ ($\lambda=0.1$) & $10.9986$ $(\lambda=0.46)$ & {min RMSE over $\boldsymbol{Z}_{\rm t}$} \\
    {\bf$\mathcal{M}_{3c}$} & $0.0557$ ($\lambda=0.1$) & $0.0992$ ($\lambda=0.1$) & $3.7285$ $(\lambda=0.54)$ & {min RMSE over $\boldsymbol{Z}_{\rm v}$} \\
    {\bf$\mathcal{M}_{3d}$} & $\gg 10^2$ ($\lambda=0.4$) & $0.4592$ ($\lambda=0.9$) & $173.949$ $(\lambda=0.98)$ & {max RMSE over $\boldsymbol{Z}_{\rm v}$} \\ \hline
  \end{tabular}}
\end{table}

\subsection{Simulated Example 2} 
\label{sec:results_ex3}

This example uses the following system \cite{jakubek2008}:
\begin{align} 
\label{eq:ex3_sys}
  w(k) =& \tan^{-1}\bigg( 1.7826w(k-1) -0.8187w(k-2) + 0.01867u(k-1) + \nonumber \\ &  \ \ \ \ \ \ \ \  \ + 0.01746u(k-2) \bigg) ,\nonumber \\
  y(k) =& w(k) + e(k),
\end{align}
where the variables are defined as before.


Four datasets were obtained from (\ref{eq:ex3_sys}), in similar fashion as for Example~1 (Sec.~\ref{sec:piroddi}). The training dataset $\boldsymbol{Z}_{\rm d}$ ($N_{\rm d} = 1700$) was obtained for $u \sim {\rm WGN}(0,\,0.02)$. A white gaussian noise $e \sim {\rm WGN}(0,\,0.01\sigma_w^2)$ was added to the output. The \emph{test} dataset $\boldsymbol{Z}_{\rm t}$ was obtained in the same way, but with $N_{\rm t} = 300$ samples.
The noise-free \emph{validation} dataset $\boldsymbol{Z}_{\rm v}$ has $N_{\rm v} = 2000$. The noisy ${\rm WGN}(0,\,0.01\sigma_{\bar{w}^2})$
static data $\boldsymbol{Z}_{\rm s}$ is presented in Figure~\ref{fig:ex3_data_wiener_static} which also indicates the operating range of $\boldsymbol{Z}_{\rm d}$, $\boldsymbol{Z}_{\rm t}$ and~$\boldsymbol{Z}_{\rm v}$.


Data from \eqref{eq:ex3_sys} were used to train the following MLP structure 
\begin{align} \label{eq:ex3_m0}
  y(k) & = \theta_1 + \theta_2 \tanh \bigg( \theta_3 + \theta_4 y(k-1) + \theta_5 y(k-2) + \theta_6 u(k-1) + \theta_7 u(k-2) \bigg) .
\end{align}

For the sake of comparison, three models were estimated: $\mathcal{M}_0$, $\mathcal{M}_2$ and $\mathcal{M}_3$. As a reference for comparison, the black-box $\mathcal{M}_0$ uses the BP with the Levenberg-Marquardt training algorithm over $\boldsymbol{Z}_{\rm d}$ only. $\mathcal{M}_2$ implements the grey-box procedure available in~\cite{barbosa2011}, where the parameters are estimated by minimizing the bi-objective problem~\eqref{eq:cost_abreu2012} and~\eqref{eq:cost_abreu2012s}\footnote{Unlike~\cite{barbosa2011}, that uses also the simulation error to fit the parameters, in this work the prediction error~\eqref{eq:cost_abreu2012} is used in all cases.}.
The estimation of the fixed point of the \emph{model} at each point in $\boldsymbol{Z}_{\rm s}$ is necessary for evaluating $J_{\rm s}$~\eqref{eq:cost_abreu2012s}. In the present example, this is done by simulating the system for 15 steps ahead for each fixed point. Due to the nonconvexity of the problem \cite{barbosa2011}, a genetic algorithm is used for training, where the black-box parameters ($\mathcal{M}_0$) are included in the initial \emph{population} of solutions and the number of generations are limited to 21. $\mathcal{M}_3$ is trained using the proposed methodology (Sec.~\ref{sec:methods}), with weighted BP~\eqref{eq:MLP_errorTerm} and the Levenberg-Marquardt algorithm.


Free-run simulation over validation dataset $\boldsymbol{Z}_{\rm v}$ is shown in Figure~\ref{fig:ex2_data_jakubek_results}, where $\mathcal{M}_0$ achieved poor results and $(\mathcal{M}_{2c},\,\mathcal{M}_{3c})$ reached better performance as expected due to the grey-box approach. The RMSE values are shown in Table~\ref{tab:ex_results_rmse}, where $\mathcal{M}_{3c}$ performed slightly better than $\mathcal{M}_{2c}$.

The difference between $\mathcal{M}_2$ and $\mathcal{M}_3$ procedures is stressed in Figure~\ref{fig:ex2_data_jakubek_trainTime} that shows the trade-off between the \emph{training time}\footnote{The computational processing time was performed in a PC with processor Intel(R) Core(TM) i7-2860QM CPU @ 2.50GHz, 8GB DDR3 1333MHz running Matlab(R) software. The parameters of $\mathcal{M}_2$ were estimated running a parallel pool with 8 workers.} and the \emph{model performance}\footnote{The \emph{performance} of the model is measured by the RMSE of the free-run simulation over the validation dataset.}. The increase in the dispersion of RMSE values for $\mathcal{M}_3$ when compared to that of $\mathcal{M}_2$ is generously compensated by computing time which is about three orders of magnitude shorter. Comparing the best case of each, the proposed procedure yielded lower RMSE in this example.

\begin{rem}
  It is worth pointing out that the best RMSE achieved by $\mathcal{M}_2$ can be improved (e.g. increasing the maximum of generations or some other parameter in the genetic algorithm), but the objective here is to emphasize the difference between the \emph{computational cost} of both grey-box procedures. The estimation of the fixed points of \emph{models} $\mathcal{M}_2$ takes around 5 seconds for evaluating the objective function, while the procedure proposed in this paper, used to obtain $\mathcal{M}_3$, takes just a few milliseconds.
\end{rem}


\begin{figure}
 \centering
  \includegraphics[width=1.0\textwidth]{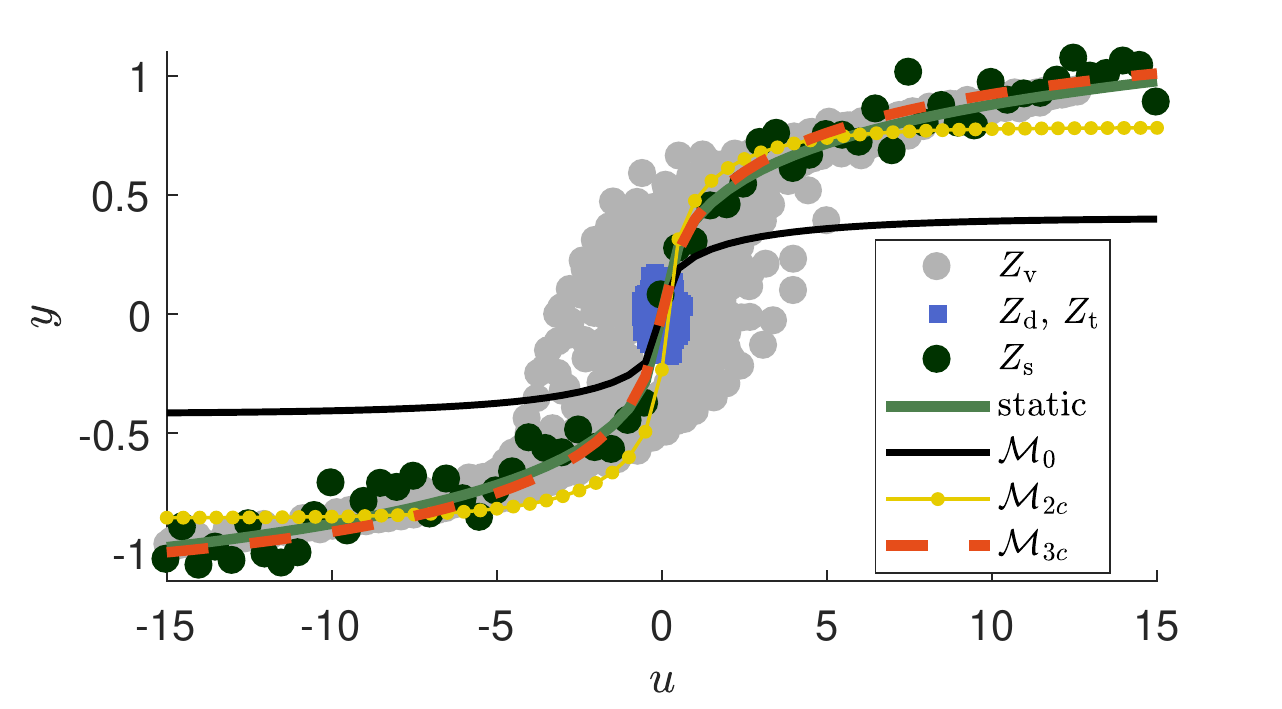}
  \caption{Comparison of the analytical static curve (thick green line) of the system (\ref{eq:ex3_sys}) with datasets ($\boldsymbol{Z}_{\rm v}$, $\boldsymbol{Z}_{\rm d}$, $\boldsymbol{Z}_{\rm t}$, $\boldsymbol{Z}_{\rm s}$) and with the estimated static behavior of the models $\mathcal{M}_0$, $\mathcal{M}_{2c}$ and $\mathcal{M}_{3c}$.}
 \label{fig:ex3_data_wiener_static}
\end{figure}

\begin{figure}
 \centering
  \includegraphics[width=1.0\textwidth]{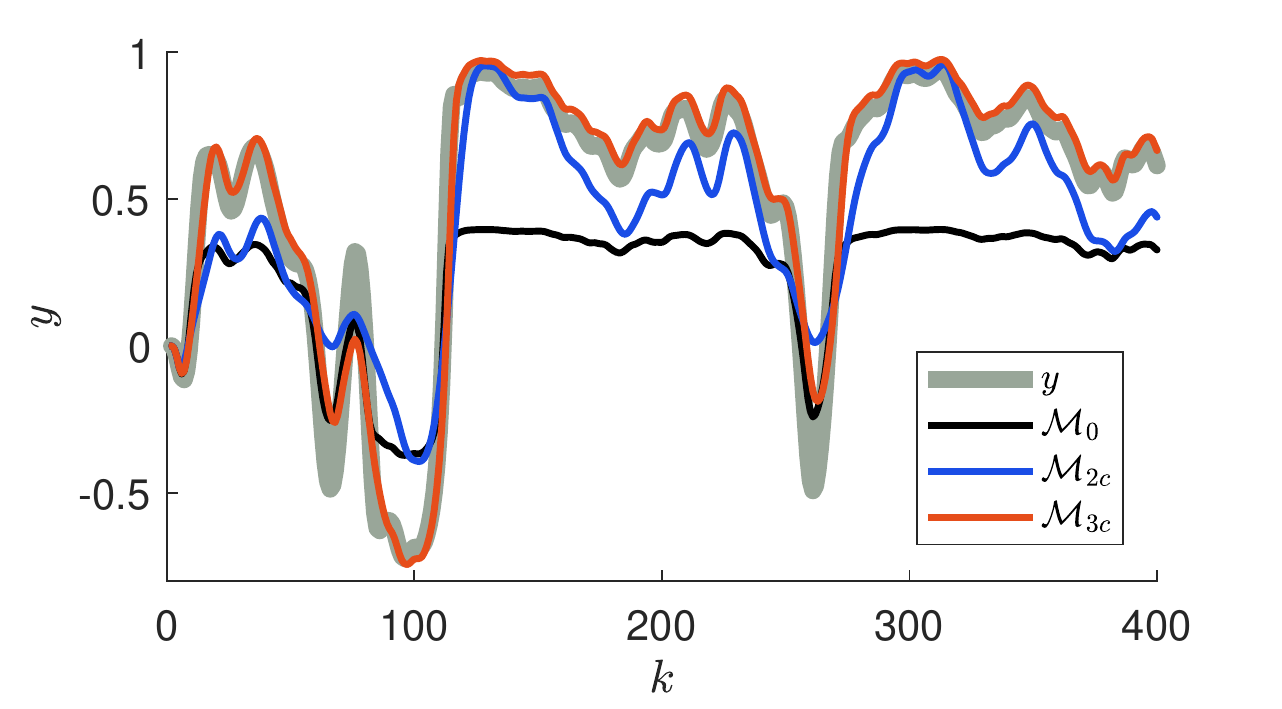}
  \caption{Free-run simulation over validation dataset $\boldsymbol{Z}_{\rm v}$ of the Simulated Example~2, where system output $y_v$ is the bold line in grey, evolutionary approach model $\mathcal{M}_2$ is denoted by the blue line and the proposed $\mathcal{M}_3$ is the red line. Model $\mathcal{M}_0$ represents the black-box approach.}
 \label{fig:ex2_data_jakubek_results}
\end{figure}

\begin{figure}
 \centering
  \includegraphics[width=.9\textwidth]{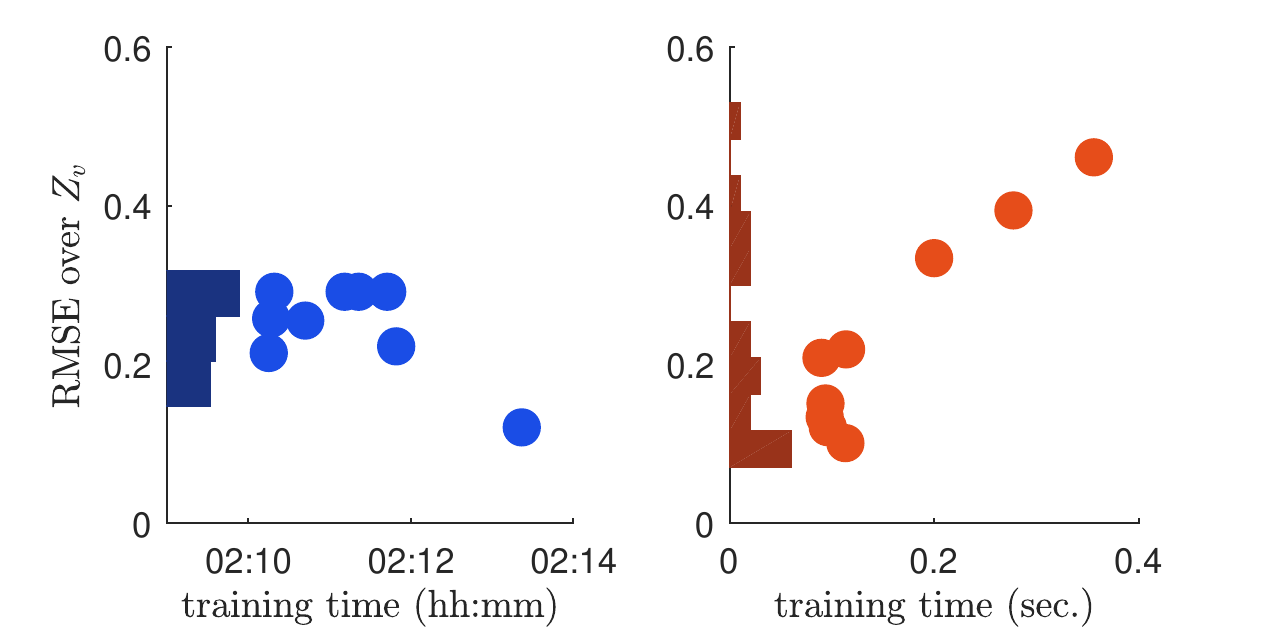}
  \caption{RMSE over validation dataset $\boldsymbol{Z}_{\rm v}$ and time for training each model $\mathcal{M}_2$ (blue) and $\mathcal{M}_3$ (red). The vertical bar graph represents the histogram of the RMSE over $\boldsymbol{Z}_{\rm v}$ and each point represents a different $\lambda$ value of the bi-objective problem.}
 \label{fig:ex2_data_jakubek_trainTime}
\end{figure}

\subsection{Experimental results: the gas-lifted oil well}
\label{sec:results_ex4}


A strategy used to avoid the lack of information caused by the failure of the \emph{downhole pressure} gauge in deep water oil production plants is the use of soft-sensor techniques. To this end, grey-box modeling procedures were applied as described in~\cite{teixeira2014, barbosa2015ifac,aguirre2017cep}.

Models of the downhole pressure should have good dynamic response and adequate static behavior. The former is important to give information about harmful dynamic events, like \emph{severe slugging}, and the latter can provide information about how to achieve good productivity conditions to ensure long service life of the oil well. To this end, grey-box models are very convenient because \emph{auxiliary information} about static behavior can be acquired from historical data.

In this context, steady-state values were estimated from historical data specifically from (almost) stationary conditions. Figure~\ref{fig:ex4_data_br_static} shows those stationary (fixed) points ($Z_{\rm s}$) that were used as \emph{auxiliary information}.

\begin{figure}
 \centering
  \includegraphics[width=1.0\textwidth]{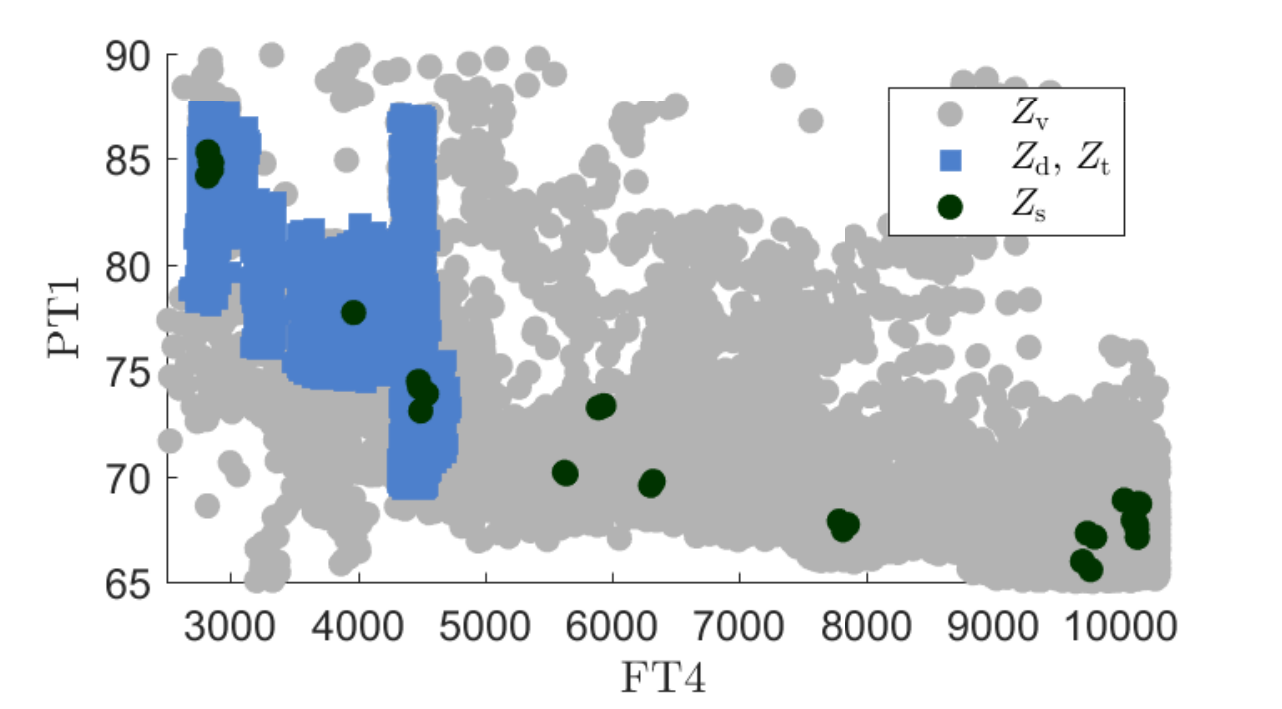}
  \caption{Comparison between the static points in $Z_{\rm s}$ and datasets $Z_{\rm d}$, $Z_{\rm t}$, $Z_{\rm v}$ represented over only one input, the instantaneous gas-lift flow rate (FT4). The output is
  the downhole pressure (PT1).}
 \label{fig:ex4_data_br_static}
\end{figure}

Figure~\ref{fig:ex4_data_br_zd_zv} shows instantaneous gas-lift flow rate (${\rm FT4}, u_1$) and downhole pressure (${\rm PT1}, y$) over the \emph{training} and \emph{validation} datasets. Like the previous numerical examples, the training and test datasets have information over a limited operating range, while the \emph{validation} dataset has operating ranges not present in the training data.

\begin{figure}
 \centering
  \includegraphics[width=1.0\textwidth]{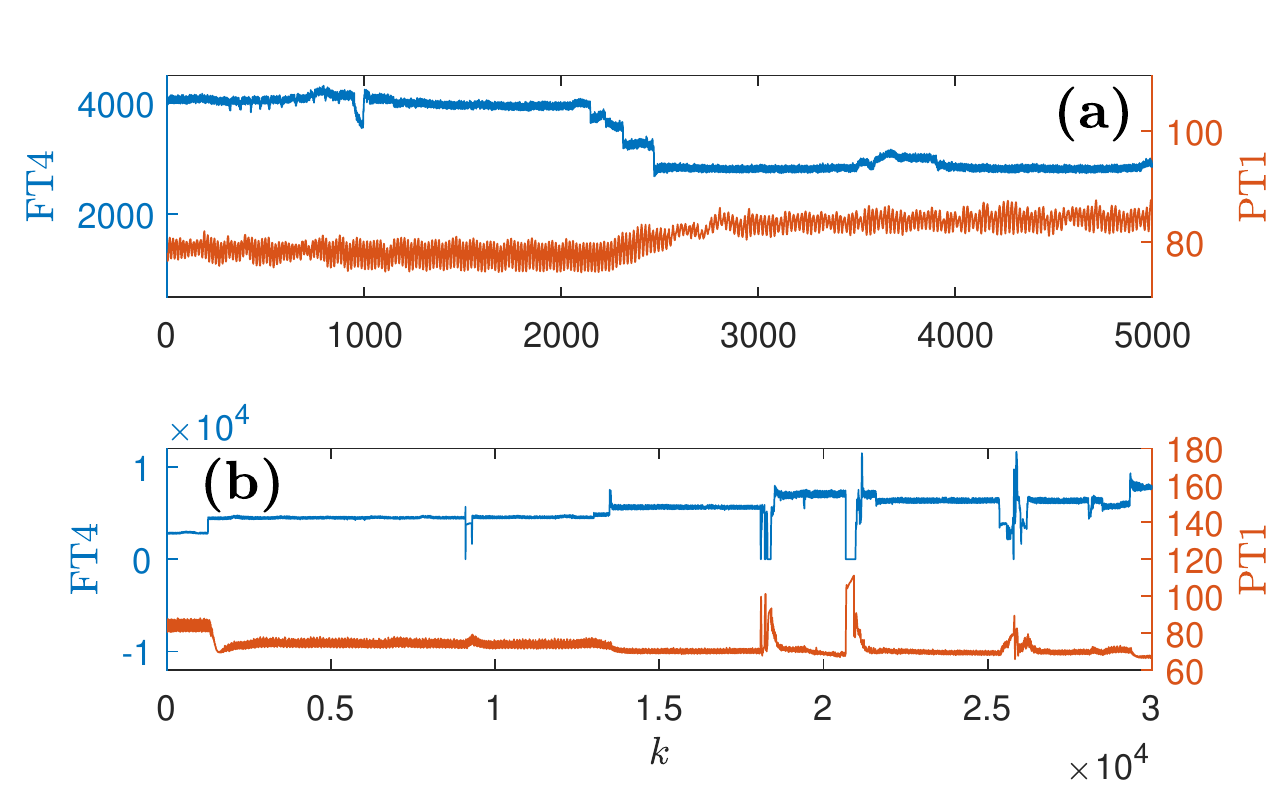}
  \caption{Instantaneous gas-lift flow rate FT4 ($u_1$) and the downhole pressure PT1 ($y$) from (a) training $\boldsymbol{Z}_{\rm d}$; and (b) validation $\boldsymbol{Z}_{\rm v}$ datasets. The fast oscillations are due to severe slugging.}
 \label{fig:ex4_data_br_zd_zv}
\end{figure}



The identified dynamic models use only platform variables 
\cite{teixeira2014,aguirre2017cep} with fixed MLP structure
%
\begin{align} \label{eq:ex4_structure}
  y(k) = \theta_0 + \sum_{i=1}^{10} \theta_i \tanh \bigg( & \theta_{i,0} + \theta_{i,1} y(k-1) + \theta_{i,2} y(k-2) + \theta_{i,3} y(k-3) \nonumber \\
       & + \theta_{i,4} u_1(k-1) + \theta_{i,5} u_1(k-42) + \theta_{i,6} u_1(k-136) \nonumber \\
       & + \theta_{i,7} u_2(k-1) + \theta_{i,8} u_2(k-42) + \theta_{i,9} u_2(k-136) \nonumber \\
       & + \theta_{i,10} u_3(k-1) + \theta_{i,11} u_3(k-5) + \theta_{i,12} u_3(k-22) \nonumber \\
       & + \theta_{i,13} u_4(k-1) + \theta_{i,14} u_4(k-5) + \theta_{i,15} u_4(k-22) \nonumber \\
       & + \theta_{i,16} u_5(k-1) + \theta_{i,17} u_5(k-5) + \theta_{i,18} u_5(k-22) \bigg),
\end{align}

\noindent that has 10 hidden nodes with activation function $\tanh(\cdot)$, and linear function in the output node. In~\eqref{eq:ex4_structure}, the signals $u_i(k)$ are variables available at the platform, and $y(k)$ is the downhole pressure. As in the previous example, the parameters of $\mathcal{M}_0$ were estimated with BP method and Levenberg-Marquardt algorithm (black-box approach). The proposed procedure is applied in training model family $\mathcal{M}_3$, with the weighted BP~\eqref{eq:weightMatrix_WLS} and the Levenberg-Marquardt algorithm, for 49 values of $\lambda\in [0.02,\,0.98]$. Models $\mathcal{M}_2$ were not implemented due to their high computational demand. On the other hand, it took only about 105 seconds to estimate all the 201 parameters of models in $\mathcal{M}_3$.

Table~\ref{tab:ex_results_rmse} shows the RMSE over the validation dataset. The best model obtained by the proposed procedure $\mathcal{M}_{3c}$ achieved improved results (Figure~\ref{fig:ex4_data_br_results_1}).  $\mathcal{M}_{3c}$ reached better performance than $\mathcal{M}_0$, especially at operating points for which
the only source of information was the \emph{auxiliary} data (e.g. $y\approx 70$). This is very relevant for many practical situations where the available dynamical data does not cover all operating regimes of the system. Obtaining static data from historical is normally a straightforward task that may help to find more representative models as shown in this real system example. Such models may provide, for instance, reliable state predictions for designing model predictive controllers, as discussed in \cite{ALTAN2020106548,WU202074}.

Table~\ref{tab:ex_results_rmse} shows the RMSE over the validation dataset for all examples. It is relevant
to point out that for none of the examples, the best models were obtained for $\lambda=0$ (which would mean to say that there was no gain in using static data). In particular, for the downhole soft-sensor, the relative importance of dynamical and static data is very well balanced. Hence it seems fair to conclude that the use of Proposition~1 makes good use of steady-state information while keeping computational costs quite moderate.

\begin{figure}
 \centering
  \includegraphics[width=1.0\textwidth]{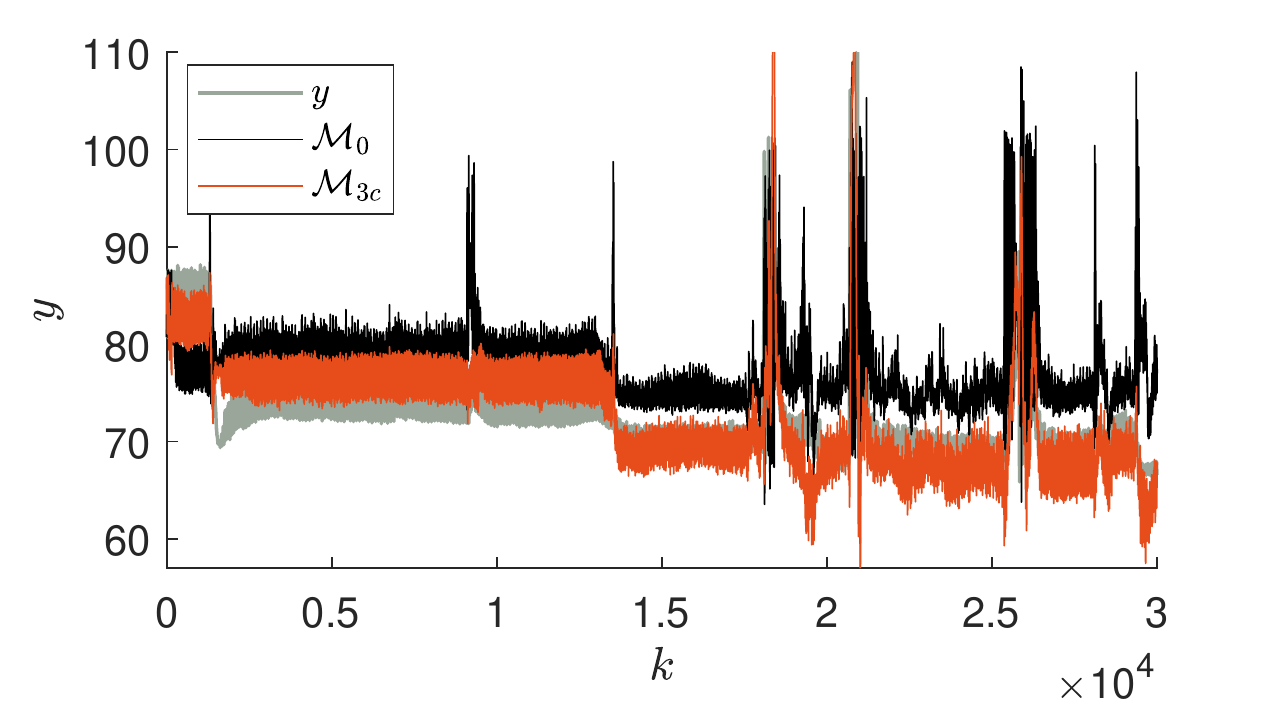}
  \caption{Free-run simulation over validation dataset $\boldsymbol{Z}_{\rm v}$.}
 \label{fig:ex4_data_br_results_1}
\end{figure}

\section{Conclusions}
\label{sec:conclusions}

This work proposes a novel method for including \emph{auxiliary information} about steady-state behavior of the system in the parameter estimation stage, by adding a new objective function (weighted problem). It was shown that the main difference between the proposed method and previous grey-box procedures is that the model fixed-points are not estimated during the identification task. {In practice, this results in computational times that are two or three orders of magnitude shorter than if the fixed-points had to be computed, as required in current grey-box procedures.} Thus, the main advantage lies in its simplicity that allows the insertion of that type of \emph{auxiliary information} in rather complex structures, such as neural networks, and the parameter estimation with ordinary algorithms, as the weighted least-squares (polynomial models) and weighted backpropagation (MLPs models). 

This work opens the possibility of applying the grey-box procedure to models with other structures, like \emph{auto regressive integrated moving average} (ARIMA), wavenets, \emph{radial basis functions} (RBF), \emph{smoothing spline models} (SSM), \emph{fuzzy models}, and others.

Simulated and experimental results showed the main aspects of the procedure and its capacity to achieve good performance with polynomial and MLP model structures.

\vspace{6pt} 




\section*{Acknowledgments}

This work was supported by  CNPq: 303412/2019-4 (LAA); CNPq: 304201/2018-9 (BHGB), FAPEMIG, IFMG and UFLA. The authors gratefully acknowledge Petrobras for providing the data. Leandro Freitas is grateful to IFMG Campus Betim for an academic leave.

\bibliographystyle{elsarticle-harv}
\bibliography{References}





\end{document}